\documentclass[11pt]{article}
\usepackage[letterpaper,margin=1in]{geometry}

\usepackage{amssymb,amsmath,color}
\usepackage{dsfont}
\usepackage{url}
\usepackage{pifont}

\usepackage{stmaryrd}

\usepackage{graphicx}
\usepackage{subfigure}

\usepackage{algorithm}
\usepackage{algorithmic}

\usepackage{hyperref}

\usepackage{amsthm}
\theoremstyle{plain}
\newtheorem{lemma}{Lemma}
\newtheorem{theorem}{Theorem}

\newtheorem{corollary}{Corollary}
\theoremstyle{definition}

\newcommand{\Sec}[1]{Section~\ref{sec:#1}}
\newcommand{\Eq}[1]{Eq.~(\ref{eq:#1})}
\newcommand{\Alg}[1]{Alg.~(\ref{alg:#1})}

\newcommand{\Theorem}[1]{Theorem~\ref{th:#1}}
\newcommand{\Corollary}[1]{Corollary~\ref{cor:#1}}
\newcommand{\Lemma}[1]{Lemma~\ref{lem:#1}}

\newcommand{\BEAS}{\begin{eqnarray*}}
\newcommand{\EEAS}{\end{eqnarray*}}
\newcommand{\BEA}{\begin{eqnarray}}
\newcommand{\EEA}{\end{eqnarray}}
\newcommand{\BEQ}{\begin{equation}}
\newcommand{\EEQ}{\end{equation}}
\newcommand{\BIT}{\begin{itemize}}
\newcommand{\EIT}{\end{itemize}}
\newcommand{\BNUM}{\begin{enumerate}}
\newcommand{\ENUM}{\end{enumerate}}
\newcommand{\BA}{\begin{array}}
\newcommand{\EA}{\end{array}}

\newcommand{\one}{\mathds{1}}

\newcommand{\argmin}{\mathop{\rm argmin}}

\def \cA{{\mathcal A}}

\def \cT{{\mathcal T}}
\def \cQ{{\mathcal Q}}
\def \cF{{\mathcal F}}

\def \E{{\mathbb E}}
\def \P{{\mathbb P}}
\def \R{{\mathbb R}}
\def \N{{\mathbb N}}

\newcommand{\Exp}[1]{\E\left[#1\right]}
\newcommand{\ExpS}[1]{\E[#1]}

\newcommand{\Prob}[1]{\P\left(#1\right)}

\newcommand{\ie}{i.e.\ }
\newcommand{\eg}{e.g.\ }

\newcommand{\Suniv}{S^{\mbox{\footnotesize \textsc{univ}}}}
\newcommand{\SPRS}{S^{\mbox{\footnotesize \textsc{prs}}}}
\newcommand{\SSPRS}{S^{\mbox{\footnotesize \textsc{sprs}}}}
\newcommand{\LBtime}[1]{\cT\left(#1, p\right)}
\newcommand{\UBtime}[1]{\widetilde{\cT}\left(#1, p\right)}
\newcommand{\comb}[2]{\mbox{\textsc{SimParallel}}(#1, #2)}

\usepackage[utf8]{inputenc} 
\usepackage[T1]{fontenc}    
\usepackage{hyperref}       
\usepackage{url}            
\usepackage{booktabs}       
\usepackage{amsfonts}       
\usepackage{nicefrac}       
\usepackage{microtype}      
\usepackage{xcolor}         
\usepackage{authblk}

\title{Black-box Acceleration of Las Vegas Algorithms\\and Algorithmic Reverse Jensen's Inequalities}

\author[1]{
Kevin Scaman\\
\small Inria Paris - Département d’informatique de l’ENS, PSL Research University
}
\date{}

\begin{document}

\maketitle

\begin{abstract}
Let $\cA$ be a Las Vegas algorithm, \ie an algorithm whose running time $T$ is a random variable drawn according to a certain probability distribution $p$.
In 1993, Luby, Sinclair and Zuckerman \cite{luby1993optimal} proved that a simple universal restart strategy can, for any probability distribution $p$, provide an algorithm executing $\cA$ and whose expected running time is $O(\ell^\star_p\log\ell^\star_p)$, where $\ell^\star_p = \Theta\left( \inf_{q\in (0,1]} Q_p(q)/q\right)$ is the minimum expected running time achievable with full prior knowledge of the probability distribution $p$, and $Q_p(q)$ is the $q$-quantile of $p$. Moreover, the authors showed that the logarithmic term could not be removed for universal restart strategies and was, in a certain sense, optimal.
In this work, we show that, quite surprisingly, $O(\ell^\star_p\log\ell^\star_p)$ is not optimal for universal restart strategies, and the logarithmic term can be replaced by a smaller quantity, thus reducing the expected running time in practical settings of interest. More precisely, we propose a novel restart strategy that executes $\cA$ and whose expected running time is $O\big(\inf_{q\in (0,1]}\frac{Q_p(q)}{q}\,\,\psi\big(\log Q_p(q),\,\log (1/q)\big)\big)$ where $\psi(a,b) = 1 + \min\left\{a + b, \,a\log^2 a, \,b \log^2 b\right\}$. This quantity is, up to a multiplicative factor, better than: 1) the universal restart strategy of \cite{luby1993optimal}, especially in low variance settings where $T$ is almost deterministic, 2) any $q$-quantile of $p$ for $q\in(0,1]$, 3) the original algorithm $\cA$ in $\Exp{T}$, and 4) quantities of the form $\phi^{-1}(\Exp{\phi(T)})$ for a large class of concave functions $\phi$. The latter extend the recent restart strategy of \cite{zamir2022wrong} achieving $O\left(e^{\Exp{\ln(T)}}\right)$, and can be thought of as algorithmic \emph{reverse Jensen's inequalities}, as $\cA$ can be executed (via restarts) in expected running time $O(\phi^{-1}(\Exp{\phi(T)}))$ instead of $\Exp{T}$. Finally, we show that the behavior of $\frac{t\phi''(t)}{\phi'(t)}$ at infinity controls the existence of reverse Jensen's inequalities by providing a necessary and a sufficient condition for these inequalities to hold.
\end{abstract}

\section{Introduction}

\begin{algorithm}[t]
\caption{Multi-start algorithm $\tilde{\cA}^S$}\label{alg:restart}
\begin{algorithmic}
\REQUIRE $x$, algorithm $\cA$ and cutoff sequence $S\in[1,+\infty)^\N$
\STATE $k\gets 0$
\REPEAT
\STATE {Run $\cA(x)$ for at most a time $S_k$}
\STATE $k\gets k+1$
\UNTIL {$\cA(x)$ terminates}
\end{algorithmic}
\end{algorithm}

Restart strategies for Las Vegas algorithms are used in multiple domains of computer science, from solving difficult optimization problems such as $k$-SAT or CSP \cite{doi:10.1137/120868177,10.1145/3313276.3316359,DBLP:conf/focs/Scheder21} to searching in large graphs \cite{10.5555/646307.687903}. Their efficiency was extensively analyzed, both theoretically \cite{Guibas96,luby1993optimal,zamir2022wrong} and in practice  \cite{DBLP:conf/uai/HoosS98,DBLP:conf/aaai/StreeterGS07a,PhysRevLett.116.170601,Evans_2020}.
More specifically, \cite{Guibas96} provided high-probability upper bounds on the running time of restart strategies, while \cite{luby1993optimal} focused their analysis on the expected running time, and showed that a simple universal restart strategy could achieve an expected running time in $O(\ell^\star_p\log\ell^\star_p)$, where $\ell^\star_p$ is the minimum expected running time achievable with full prior knowledge of the probability distribution of the running time $p$.
Unfortunately, the logarithmic term can sometimes lead to situations in which the original algorithm is faster than the restart strategy, for example for deterministic running times for which $\ell^\star_p = \Exp{T}$.
More recently, \cite{zamir2022wrong} showed that another restart strategy could achieve an expected running time in $O\left(e^{\Exp{\ln(T)}}\right)$, thus always improving on the original algorithm due to Jensen's inequality. This upper bound is however not comparable to $O(\ell^\star_p\log\ell^\star_p)$, and can sometimes be much larger, for example when $T=+\infty$ with non-negligible probability $p\in(0,1)$, in which case $e^{\Exp{\ln(T)}}=+\infty$ while $\ell^\star_p$ remains finite.
Nonetheless, the existence of a restart strategy that always improves upon the original algorithm is quite surprising, as \cite{luby1993optimal} showed in their work that $O(\ell^\star_p\log\ell^\star_p)$ was, in a certain sense, optimal (see \Theorem{luby}). This issue raises the question of optimality of the universal strategy of \cite{luby1993optimal}, and calls for the design of algorithms improving on both quantities.

In this work, we show that such an algorithm exists and, quite surprisingly, the logarithmic term in $O(\ell^\star_p\log\ell^\star_p)$ can be replaced by a smaller quantity, thus reducing the expected running time in practical settings of interest. More precisely, in \Sec{quantiles}, we first show that a simple restart strategy can reach any fixed quantile (up to a constant multiplicative term), and propose in \Sec{adaptive} a novel restart strategy, called \emph{symmetric parallel restart strategy} (SPRS), that executes $\cA$ and whose expected running time is $O\big(\inf_{q\in (0,1]}\frac{Q_p(q)}{q}\,\,\psi\big(\log Q_p(q),\,\log (1/q)\big)\big)$ where $\psi(a,b) = 1 + \min\left\{a + b, \,a\log^2 a, \,b \log^2 b\right\}$.
This strategy is based on the a novel framework, called \emph{simulated parallel computations}, that allows to run multiple (and even countably infinite) restart strategies at the same time by simulating parallel execution on a single computing unit (see \Lemma{combine}).
Moreover, we investigate in \Sec{jensen} the relationship between our expected running time $O\big(\inf_{q\in (0,1]}\frac{Q_p(q)}{q}\,\,\psi\big(\log Q_p(q),\,\log (1/q)\big)\big)$ and quantities of the form $O(\phi^{-1}(\Exp{\phi(T)}))$ for positive, increasing and concave functions $\phi$, and show that the behavior of $\frac{t\phi''(t)}{\phi'(t)}$ at infinity drives the existence or absence of algorithmic reverse Jensen's inequalities. More precisely, our results are twofold:
\BNUM
\item \textbf{Sufficient condition:} If $\liminf_{t\to +\infty}\frac{t\phi''(t)}{\phi'(t)} > -2$, then the expected running time of SPRS, for any distribution $p$, is upper bounded by $\LBtime{\SSPRS} = O(\phi^{-1}(\Exp{\phi(T)}))$ and the reverse Jensen's inequality holds in an algorithmic sense.
\item \textbf{Necessary condition:} If $\limsup_{t\to +\infty}\frac{t\phi''(t)}{\phi'(t)} < -2$, then no restart strategy can achieve an expected running time $\LBtime{S} = O(\phi^{-1}(\Exp{\phi(T)}))$.
\ENUM

\paragraph{Notations}
In what follows, we use $f(x) = O(g(x))$ (resp. $f(x) = \Omega(g(x))$) as a shorthand for $\exists C > 0$ a universal constant such that $\forall x, f(x) \leq C\,g(x)$ (resp. $f(x) \geq C\,g(x)$), and $f(x) = \Theta(g(x))$ as a shorthand for both $f(x) = O(g(x))$ and $f(x) = \Omega(g(x))$.
Finally, $\log(x)$ denotes the logarithm of $x$ in base $2$, while $\ln(x)$ denotes the natural logarithm of $x$.

\section{Restart strategies: setup and notations}

Let $\cA$ be a Las Vegas algorithm, and $T\in[1,+\infty]$ its random running time. We assume that $T$ is lower bounded by $1$ (\eg a single iteration, step or computation) and may take infinite values when the algorithm does not terminate. Moreover, we denote as $p$ its probability distribution over $[1,+\infty]$. Note that $T$ is usually an integer, for example when it represents the number of iterations of an optimization algorithm, but our analysis also extends to the more general setting in which $T$ is a real value. 
Our goal is to design a new algorithm $\tilde{\cA}$ whose output is always that of $\cA$, but its expected running time is smaller. More specifically, we consider restart schemes of the following form: let $S\in[1,+\infty)^\N$ be a (possibly random) sequence of \emph{cutoff times}, \ie times at which the original algorithm $\cA$ is stopped and restarted. We then construct a \emph{multi-start} algorithm $\tilde{\cA}^S$ by, for a given input $x$, sequentially running $\cA(x)$ multiple times with cutoff times $S_k$ for $k\geq 0$, until the algorithm eventually terminates before the cutoff (see \Alg{restart}). Let $(T_k)_{k\geq 0}$ be the i.i.d running times of algorithm $\cA$ at each iteration $k\geq 0$ of \Alg{restart}. We denote as $K^\star = \min\{k\geq 0~:~T_k \leq S_k\}$ the last iteration of $\tilde{\cA}^S$ (\ie the iteration in which $\cA$ terminated and provided the desired output), and $\LBtime{S} = \Exp{\sum_{k=0}^{K^\star-1} S_k + T_{K^\star}}$ the expected running time of $\tilde{\cA}^S$. Our objective is thus to find a sequence $S$ minimizing the expected running time of the multi-start strategy.
In their seminal paper, \cite{luby1993optimal} showed that this minimal expected running time is
\BEQ\label{eq:opt}
\inf_{S\in[1,+\infty)^\N} \LBtime{S} \quad=\quad \inf_{\alpha\geq 1} \frac{\Exp{\min\{T,\,\alpha\}}}{\Prob{T\leq \alpha}}\,,
\EEQ
and is achieved for a constant sequence $S = (\alpha^\star,\alpha^\star,\dots)$ where $\alpha^\star$ minimizes \Eq{opt}.
The right-hand side of \Eq{opt} is tightly connected to \emph{quantiles} of the distribution $Q_p(q) = \min\{\alpha\in[1,+\infty]~:~\Prob{T\leq \alpha}\geq q\}$ for $q\in(0,1]$, \ie the first time at which the algorithm $\cA$ has already terminated with probability at least $q$.
\begin{lemma}[adapted from \cite{luby1993optimal}]\label{lem:opt}
For any distribution $p$, we have
\BEQ\label{eq:optimal}
\frac{1}{2}\,\inf_{q\in(0,1]} \frac{Q_p(q)}{q} \quad\leq\quad \inf_{S\in[1,+\infty)^\N} \LBtime{S} \quad\leq\quad \inf_{q\in(0,1]} \frac{Q_p(q)}{q}\,.
\EEQ
\end{lemma}
\begin{proof}
Let $\ell^\star_p = \inf_{\alpha\geq 1} \Exp{\min\{T,\,\alpha\}} / \Prob{T\leq \alpha}$ and $L^\star_p = \inf_{\alpha\geq 1} \alpha / \Prob{T\leq \alpha}$. First, note that $L^\star_p = \inf_{q\in(0,1]} Q_p(q)/q$ by definition of the quantiles, as
\BEQ
L^\star_p \quad=\quad \inf_{\alpha\geq 1} \frac{\alpha}{\Prob{T\leq \alpha}} \quad=\quad \inf_{\alpha\geq 1} \frac{Q_p(\Prob{T\leq \alpha})}{\Prob{T\leq \alpha}} \quad=\quad \inf_{q\in\cQ} \frac{Q_p(q)}{q}\,,
\EEQ
where $\cQ = \left\{q\in(0,1]~:~\exists \alpha\geq 1, \Prob{T\leq \alpha} = q\right\}$.
However, the minimum of $Q_p(q)/q$ over $q\in(0,1]$ is necessarily attained in $\cQ$, as $\forall q\in(0,1]$, $Q_p(q) = Q_p(\Prob{T\leq Q_p(q)})$ and $\Prob{T\leq Q_p(q)} \geq q$.
Second, by definition of $\ell^\star_p$, we directly have $\Exp{\min\{T,\,\alpha\}} \leq \alpha$ and thus $\ell^\star_p \leq L^\star_p$. For the second inequality, we can lower bound $\Exp{\min\{T,\,\alpha\}}/\Prob{T\leq \alpha}$ depending on the value of $\alpha$: if $\alpha \leq 2\Exp{T}$, then $\Exp{\min\{T,\,\alpha\}} / \Prob{T\leq \alpha} \geq \alpha / 2\Prob{T\leq \alpha} \geq L^\star_p / 2$, and if $\alpha > 2\Exp{T}$, then $\Exp{\min\{T,\,\alpha\}} / \Prob{T\leq \alpha} \geq 2\Exp{T} / \Prob{T\leq \alpha} \geq 2\Exp{T} \geq L^\star_p / 2$, where the last inequality is due to Markov's inequality, as $L^\star_p\leq 2\Exp{T} / \Prob{T\leq 2\Exp{T}} \leq 4\Exp{T}$. As a consequence, we have that $L^\star_p \geq \ell^\star_p \geq L^\star_p /2$ (note that the factor $1/2$ is slightly better than $1/4$ obtained in \cite{luby1993optimal}).
\end{proof}

We thus have $\LBtime{S} = \Theta\left(\inf_{q\in(0,1]} \frac{Q_p(q)}{q}\right)$, and intuitively the optimal restart strategy can be as good as any $q$-quantile for $q\in(0,1]$, which may significantly reduce the expected running time compared to $\Exp{T}$.
Unfortunately, such a restart strategy requires to know beforehand the distribution of the running time $p$ (to find the optimal $\alpha^\star$), and is thus impossible to use in settings in which the distribution is unknown.
Fortunately, \cite{luby1993optimal} showed that a simple universal cutoff sequence achieves a near optimal expected running time without any knowledge on $p$.
\begin{theorem}[from \cite{luby1993optimal}]\label{th:luby}
There exists a sequence $\Suniv\in[1,+\infty)^\N$ such that, for any distribution $p$,
\BEQ
\LBtime{\Suniv} \leq 192\, \ell^\star_p (\log \ell^\star_p + 5)\,,
\EEQ
where $\ell^\star_p = \inf_{S\in[1,+\infty)^\N} \LBtime{S}$. Moreover, for any fixed sequence $S\in[1,+\infty)^\N$,
\BEQ
\max_{p~:~\ell^\star_p = \ell} \LBtime{S} \geq \frac{1}{8}\,\ell \log \ell\,.
\EEQ
\end{theorem}

In other words, this universal restart strategy achieves an expected running time in $\LBtime{\Suniv} = O\left(\inf_{q\in(0,1]} \frac{Q_p(q)}{q}\left(1+\log\left( \frac{Q_p(q)}{q} \right)\right)\right)$, and any upper bound on the expected running time of a universal restart strategy depending solely on $\ell^\star_p$ will exhibit a multiplicative logarithmic term. However, this negative result does not imply that smaller upper bounds cannot be achieved with other characteristics of the distribution $p$, as will be shown in \Sec{adaptive}.

\section{A simple strategy for fixed quantiles}\label{sec:quantiles}
First, we show in this section that a simple restart sequence can reach an expected running time upper bounded by any predefined quantile. More specifically, we provide a slightly stronger result that upper bounds $\UBtime{S} = \Exp{\sum_{k=0}^{K^\star} S_k}$ instead of $\LBtime{S}$. The quantity $\UBtime{S}$ is the expected sum of all cutoff times before termination, and is thus always greater than $\LBtime{S}$. 

\begin{theorem}\label{th:quantiles}
Let $q\in(0,1]$ and $S^q$ be a sequence of restart times such that, $\forall k\geq 0$, $S^q_k=\left(1 - \frac{q}{2}\right)^{-k}$. Then, for any distribution $p$, we have
\BEQ
\UBtime{S^q} \quad\leq\quad 4\,\,\frac{Q_p(q)}{q}\,.
\EEQ
\end{theorem}
\begin{proof}
By definition of $\UBtime{S}$, we have, for any fixed sequence $S\in[1,+\infty)^\N$,
\BEQ
\UBtime{S} \quad=\quad \Exp{ \sum_{k=0}^{+\infty} \one\{k \leq K^\star\} S_k} \quad=\quad \sum_{k=0}^{+\infty} \Prob{K^\star \geq k} S_k \quad=\quad \sum_{k=0}^{+\infty} \prod_{l=0}^{k-1} \Prob{T > S_l} S_k\,,
\EEQ
as the running times $T_l$ are i.i.d. random variables. Let $S_k = \gamma^k$ for $\gamma > 1$ and $K_q = \lceil\log_\gamma Q_p(q)\rceil$. Then, $\forall k \geq K_q$, $\Prob{T \leq S_k} \geq q$, and thus
\BEQ
\UBtime{S} \quad\leq\quad \sum_{k=0}^{K_q - 1} \gamma^k + \sum_{k=0}^{+\infty} (1-q)^k \gamma^{k+K_q} \quad\leq\quad \frac{\gamma^{K_q}}{\gamma - 1} + \frac{\gamma^{K_q}}{1 - (1-q)\gamma}\,.
\EEQ
Finally, noting that $\gamma^{K_q} \leq \gamma Q_p(q)$ and optimizing over $\gamma$ gives $\gamma = 1-\frac{q}{2}$ and the desired result.
\end{proof}

Note that the multiplicative factor of $4$ in \Theorem{quantiles} is very small compared to $192$ obtained in \cite{luby1993optimal}, which makes the upper bound more valuable in practice. For example, the restart strategy with $S^{1/2}$ returns the output of the algorithm in at most $8$ times the median of the running time. Moreover, when $q < 1$, these restart strategies are always better than the original algorithm, as applying Markov's inequality to $T$ immediately gives $(1-q)Q_p(q)\leq \Exp{T}$, and thus
\BEQ
\LBtime{S^q} \quad\leq\quad \frac{4}{q(1-q)}\,\Exp{T}\,.
\EEQ

More generally, quantiles are also better than any $\Exp{T^{a}}^{\frac{1}{a}}$ where $a>0$, as $(1-q)Q_p(q)^a\leq \Exp{T^a}$ and thus
\BEQ\label{eq:moments}
\LBtime{S^q} \quad\leq\quad \frac{4}{q(1-q)^{\frac{1}{a}}}\,\Exp{T^{a}}^{\frac{1}{a}}\,.
\EEQ

Note that, by Jensen's inequality, we have $\Exp{T^{a}}^{\frac{1}{a}} \leq \Exp{T}$, and thus \Eq{moments} can be thought of, up to a multiplicative factor, as a \emph{reverse Jensen's inequality}, as $\tilde{\cA}^{S^q}$ executes $\cA$ with a running time in $O\left(\Exp{T^{a}}^{\frac{1}{a}}\right)$ without any assumption on the running time of $\cA$ or its distribution $p$.
In \Sec{jensen}, we will see that an adaptive version of these restart strategies can provide reverse Jensen's inequalities for a large class of concave functions.

\section{Adaptive strategies through simulated parallel computations}\label{sec:adaptive}
If an optimal quantile $q^\star \in \argmin_{q\in(0,1]} \frac{Q_p(q)}{q}$ is known beforehand, then $S^{q^\star}$ is already optimal (up to a multiplicative factor). However, in general this is not the case, and we thus need to discover $q^\star$ while the algorithm is running. Our approach to such an adaptive algorithm is to run multiple restart sequences $S^{q_i}$ in \emph{parallel}, or rather simulate parallel computations by randomly updating one of the sequences at each iteration, and waiting for one of the resulting algorithms to terminate.
\begin{algorithm}[t]
\caption{$\comb{S}{r}$}\label{alg:combine}
\begin{algorithmic}
\REQUIRE cutoff sequences $(S^0,S^1,\dots) \in {[1,+\infty)^\N}^\N$ and $(r_0,r_1,\dots) \in [0,1]^\N$ such that $\sum_i r_i = 1$
\ENSURE randomly combined cutoff sequence $\tilde{S}$
\STATE For $i\in\N$, $K_i \gets 0$
\REPEAT
\STATE Sample $I\in\N$ from the probability density $\tilde{r}_i \propto r_i / S^i_{K_i}$
\STATE $\tilde{S}_k \gets S^I_{K_I}$
\STATE $K_I \gets K_I + 1$
\UNTIL
\end{algorithmic}
\end{algorithm}
To do so, we create a random sequence $\comb{S}{r}$ from base sequences $S = (S^0,S^1,\dots)$ and a probability density $r = (r_0,r_1,\dots)$ over $\N$ using the procedure described in \Alg{combine}.
Intuitively, the random sequence $\comb{S}{r}$ runs all sequences $S^i$ in \emph{parallel} by selecting, at each iteration, one of the sequences at random and incrementing its current position by one.
The random selection is chosen such that the computation speed of each algorithm $\tilde{\cA}^{S^i}$ is proportional to $r_i$.

\begin{lemma}\label{lem:combine}
Let $(S^i)_{i\in\N} \in ([1,+\infty)^\N)^\N$ be a sequence of restart sequences and $(r_i)_{i\in\N} \in [0,1]^\N$ a probability density over $\N$. Then, for any distribution $p$, we have
\BEQ
\UBtime{\comb{S}{r}} \quad\leq\quad \inf_{i\in\N} \frac{\UBtime{S^i}}{r_i}\,.
\EEQ
\end{lemma}
\begin{proof}
For any $i,t\in\N$, let $I_t\in\N$ be the index of the sequence selected at iteration $t$ of \Alg{combine}, $K_{i,t} = \sum_{j=0}^t\one\{I_j = i\}$ the number of iterations in which sequence $S^i$ was selected up to iteration $t$ (included), and $R_{i,t} = \sum_{j=0}^{K_{i,t}-1} S^i_j$ the running time of algorithm $\tilde{\cA}^{S^i}$ up to iteration $t$ (included). Moreover, let $\widetilde{R}_t = \sum_{j=0}^t S^{I_j}_{K_{I_j,j-1}}$ be the running time of the whole algorithm up to iteration $t$ (included). First, note that $\UBtime{\comb{S}{r}} = \ExpS{\widetilde{R}_{T^\star}}$, where $T^\star = \min\{t\in\N~:~\exists i\in\N, K_{i,t} \geq K^\star_i\}$ is the first iteration at which one of the parallel algorithms $\tilde{\cA}^{S^i}$ terminates, and $K^\star_i$ is the last iteration of $\tilde{\cA}^{S^i}$. In what follows we will show that, $\forall i\in\N$, $X_{i,t} = R_{i,t} - r_i \widetilde{R}_{t}$ is a martingale, and $T^\star$ is a stopping time adapted to the natural filtration of $((I_t, T_t))_{t\in\N}$. Assume these conditions hold, we have that, for any $\tau\in\N$, $\Exp{X_{i,T^\star\wedge\tau}} = 0$ using Doob's theorem and, $\forall i\in\N$,
\BEQ
\UBtime{\comb{S}{r}} \quad=\quad \lim_{\tau\to +\infty} \Exp{\widetilde{R}_{T^\star\wedge\tau}} \quad=\quad \lim_{\tau\to +\infty} \frac{\Exp{R_{i,T^\star\wedge\tau}}}{r_i} \quad\leq\quad \frac{\UBtime{S^i}}{r_i}\,,
\EEQ
where the first equality follows from monotone convergence and the last inequality follows from $K_{i,T^\star} \leq K^\star_i$, thus leading to the desired result. Let us now show that $X_{i,t}$ is a martingale: let $\cF_t = \sigma(I_0,\dots,I_t,T_0,\dots,T_t)$ be the natural filtration of $((I_t,T_t))_{t\in\N}$. Then, we have
\BEQ
\BA{lll}
\Exp{X_{i,t+1}~|~\cF_t} &=& \Exp{R_{i,t+1} - r_i \widetilde{R}_{t+1}~|~\cF_t}\\
&=& \Exp{X_{i,t} + \one\{I_{t+1}=i\}S^i_{K_{i,t}} - r_i S^{I_{t+1}}_{K_{I_{t+1},t}}~|~\cF_t}\\
&=& X_{i,t} + \frac{r_i}{Z_t S^i_{K_{i,t}}}S^i_{K_{i,t}} - r_i \sum_{j=0}^{+\infty} \frac{r_j}{Z_t S^j_{K_{j,t}}}S^j_{K_{j,t}}\\
&=& X_{i,t} + \frac{r_i}{Z_t} - \frac{r_i}{Z_t}\\
&=& X_{i,t}\,,
\EA
\EEQ
where $Z_t = \sum_{j=0}^{+\infty} \frac{r_j}{S^j_{K_{j,t}}}$, and $\Exp{|X_{i,t}|}\leq (1+r_i)\Exp{\widetilde{R}_t} = (1+r_i)\sum_{j=0}^t \Exp{Z_t^{-1}}$ is bounded, as $Z_t \geq \sum_{j=0}^{+\infty} \frac{r_j}{S^j_t} $ which is a deterministic and positive quantity. Finally, $T^\star$ is a stopping time, as, $\forall t\geq 0$, the event $\{T^\star = t\}$ (\ie the algorithm terminates at iteration $t$) only depends on $(I_0,\dots,I_t)$ and $(T_0,\dots,T_t)$ (\ie the sequences selected up to $t$, and the running times of $\cA$ up to $t$), which concludes the proof.
\end{proof}

Using this lemma, we can combine the restart strategies of \Sec{quantiles} for $q_i = 2^{-i-1}$ in order to obtain a random restart strategy that is better than any quantile (up to a multicative factor), and only a poly-logarithmic multiplicative factor from the optimal expected running time of \Eq{optimal}.

\begin{theorem}\label{th:adaptive}
Let $C>0$ and $\varphi:\N\to\R_+$ be such that $\sum_{i=0}^{+\infty} \frac{1}{\varphi(i)} \leq C$. Then, there exists a random sequence $S^\varphi = \comb{S^{2^{-i-1}}}{\frac{1}{C\varphi(i)}}$ such that, for any distribution $p$, we have
\BEQ
\LBtime{S^\varphi} \quad\leq\quad 8C\,\inf_{q\in (0,1]}\frac{Q_p(q)}{q} \,\, \varphi\left( \left\lfloor\log\left(\frac{1}{q}\right)\right\rfloor \right)\,.
\EEQ
\end{theorem}
\begin{proof}
Applying \Lemma{combine} to the sequences $S^{q_i}$ from \Theorem{quantiles} for $q_i = 2^{-i-1}$ and $r_i = 1/C\varphi(i)$, we obtain that
\BEQ
\UBtime{S^\varphi} \quad\leq\quad \inf_{i\in\N} \UBtime{S^{q_i}}C\varphi(i) \quad\leq\quad 4C\inf_{i\in\N} \frac{Q_p(q_i)}{q_i}\varphi(i)\,.
\EEQ
The desired result follows from, $\forall q\in(0,1]$, $\frac{Q_p(q)}{q} \geq \frac{Q_p(q_i)}{2q_i}$ for $i=\left\lfloor\log\left(1/q\right)\right\rfloor$.
\end{proof}

As \Theorem{adaptive} holds for any function whose inverse is integrable, $\varphi(x)$ should be asymptotically larger than $x$, and we can use any function of the form $x^{1+\varepsilon}$, $x\log^{1+\varepsilon}(1+x)$, $x\log(1+x)\log^{1+\varepsilon}(1+\log(1+x))$, ...
For simplicity, we describe in \Alg{AQ} the algorithm for $\varphi(x) = 1 + x\log^2 x$ that only displays a poly-logarithmic multiplicative term compared to $x$.

\begin{algorithm}[t]
\caption{Parallel restart strategy (simple version) $\tilde{\cA}^{\SPRS}$}\label{alg:AQ}
\begin{algorithmic}
\REQUIRE x
\STATE For $i\in\N$, $K_i \gets 0$
\REPEAT
\STATE Sample $I\in\N$ from the probability density $\tilde{r}_i \propto (1 - \frac{1}{2^{i+2}})^{K_i} / (1 + i \log^2 i)$
\STATE Run $\cA(x)$ for at most a time $(1 - \frac{1}{2^{I+2}})^{-K_I}$
\STATE $K_I \gets K_I + 1$
\UNTIL {$\cA(x)$ terminates}
\end{algorithmic}
\end{algorithm}

\begin{corollary}\label{cor:SAQ}
There exists a random sequence $\SPRS$ such that, for any distribution $p$, we have
\BEQ
\LBtime{\SPRS} \quad\leq\quad 23\,\inf_{q\in (0,1]}\frac{Q_p(q)}{q}\left(1 + \log\left(\frac{1}{q}\right) \log^2\log\left(\frac{1}{q}\right)\right)\,.
\EEQ
\end{corollary}
\begin{proof}
The result is a direct consequence of \Theorem{adaptive} and $\sum_{i=0}^{+\infty} 1/(1+i\log^2 i) \leq 2.8$.
\end{proof}

In settings in which the optimal quantile $Q_p(q^\star)$ may be large but $q^\star$ is bounded by a constant, our algorithm achieves an optimal $O(\ell^\star_p)$.
For example, when the distribution is deterministic $T=c$, we have $q^\star = 1$ and $\LBtime{\SPRS} \leq 23c$ while the upper bound for the universal strategy of \cite{luby1993optimal} displays a logarithmic multiplicative factor $\LBtime{\Suniv} = O(c(1+\log c))$.
Moreover, as $Q_p(q) \geq 1$, we have $\LBtime{\SPRS} = O(\ell^\star_p (1+\log\ell^\star_p\log^2\log\ell^\star_p))$, and thus only exhibit, in a worst-case setting, an added multiplicative factor $\log^2\log\ell^\star_p$ compared to the expected running time of the universal strategy of \cite{luby1993optimal}.
Interestingly, this sub-optimal term can be removed by using $\comb{S}{r}$ to also run the universal strategy in parallel.

\begin{theorem}\label{th:uaq}
Let $\varphi:\N\to\R_+$ be such that $\sum_{i=0}^{+\infty} \frac{1}{\varphi(i)} $ is bounded. Then, there exists a random sequence $\SSPRS$ and constant $C>0$ such that, for any distribution $p$, we have
\BEQ
\LBtime{\SSPRS} \quad\leq\quad C\,\inf_{q\in (0,1]}\frac{Q_p(q)}{q}\,\,\psi\Big(\left\lfloor\log Q_p(q)\right\rfloor,\,\left\lfloor\log\left(1/q\right)\right\rfloor\Big)\,,
\EEQ
where $\psi(a,b) = 1 + \min\left\{a + b,\, \varphi(a),\, \varphi(b)\right\}$.
\end{theorem}
\begin{proof}
Let $S=(S^0,S^1,\dots)$ and $r=(r^0,r^1,\dots)$ such that $S^0 = \Suniv$ from \cite{luby1993optimal}, $r_0=1/3$, and, for any $i\in\N$, let $S^{2i + 1} = S^{q_i}$ from \Theorem{quantiles} where $q_i=2^{-i-1}$, $S^{2i+2} = (\alpha_i,\alpha_i,\dots)$ where $\alpha_i = 2^{i+1}$ and $r_{2i+1} = r_{2i+2} = 1/3C\varphi(i)$. Then, applying \Lemma{combine} to $\SSPRS = \comb{S}{r}$ gives
\BEQ
\UBtime{\SSPRS} \quad\leq\quad 3 \min\left\{\UBtime{\Suniv},\,4C\inf_{i\in\N} \frac{Q_p(q_i)}{q_i}\varphi(i),\,C\inf_{i\in\N} \frac{\alpha_i}{\Prob{T\leq \alpha_i}}\varphi(i)\right\}\,.
\EEQ
From \cite{luby1993optimal} and \Lemma{opt}, there exists a constant $C>0$ such that
\BEQ
\UBtime{\Suniv} \leq C\inf_{i\in\N} \frac{Q_p(q_i)}{q_i}\left(1+\log\left(\frac{Q_p(q_i)}{q_i}\right)\right)\,.
\EEQ
Moreover, $\forall q\in(0,1]$, we have $\frac{Q_p(q)}{q} \geq \frac{Q_p(q_i)}{2q_i}$ for $i=\left\lfloor\log\left(1/q\right)\right\rfloor$, and thus $\inf_{i\in\N} \frac{Q_p(q_i)}{q_i}\varphi(i) \leq 2\inf_{q\in (0,1]}\frac{Q_p(q)}{q}\,\,\varphi\left(\log\left(1/q\right)\right)$. Finally, $\forall q\in(0,1]$, we have $\frac{Q_p(q)}{q} \geq \frac{\alpha_i}{2\Prob{T\leq\alpha_i}}$ for $i=\left\lfloor\log Q_p(q)\right\rfloor$, and thus $\inf_{i\in\N} \frac{\alpha_i}{\Prob{T\leq \alpha_i}}\varphi(i) \leq 2 \inf_{q\in (0,1]}\frac{Q_p(q)}{q}\,\,\varphi\left(\log Q_p(q)\right)$. Combining these three upper bounds gives the desired result.
\end{proof}

The \emph{symmetric parallel restart strategy} (SPRS) described in \Alg{UAQ} is a particular example of that of \Theorem{uaq} when $\varphi(x) = 1 + x\log^2 x$. SPRS displays an expected running time that is symmetric with respect to the quantities $Q_p(q)$ and $1/q$, and is thus efficient in both settings $Q_p(q) \ll 1/q$ and $Q_p(q) \gg 1/q$. When $Q_p(q) \approx 1/q$, the term $\log(Q_p(q)/q)$ (\ie $a+b$ in $\psi(a,b)$) becomes minimal and allows the upper bound to avoid any additional poly-logarithmic terms in this setting.

\begin{algorithm}[t]
\caption{Parallel restart strategy (symmetric version) $\tilde{\cA}^{\SSPRS}$}\label{alg:UAQ}
\begin{algorithmic}
\REQUIRE x (and the universal sequence $\Suniv$ from \cite{luby1993optimal})
\STATE $K^1 \gets 0$
\STATE For $i\in\N$, $K^2_i \gets 0$
\REPEAT
\STATE Sample $U\in\{1,2,3\}$ uniformly at random
\IF {$U = 1$}
\STATE Run $\cA(x)$ for at most a time $\Suniv_{K^1}$
\STATE $K^1 \gets K^1 + 1$
\ELSIF {$U = 2$}
\STATE Sample $I\in\N$ from the probability density $\tilde{r}_i \propto (1 - \frac{1}{2^{i+2}})^{K^2_i} / (1 + i \log^2 i)$
\STATE Run $\cA(x)$ for at most a time $(1 - \frac{1}{2^{I+2}})^{-K^2_I}$
\STATE $K^2_I \gets K^2_I + 1$
\ELSIF {$U = 3$}
\STATE Sample $I\in\N$ from the probability density $\tilde{r}_i \propto 2^{-i-1} / (1 + i \log^2 i)$
\STATE Run $\cA(x)$ for at most a time $2^{I+1}$
\ENDIF
\UNTIL {$\cA(x)$ terminates}
\end{algorithmic}
\end{algorithm}

\section{From quantiles to reverse Jensen's inequalities}\label{sec:jensen}

For any positive, continuous and increasing $\phi:[1,+\infty)\to\R_+$, let $\mu_\phi(T) = \phi^{-1}\left(\Exp{\phi(T)}\right)$. As $\phi$ is injective, the inverse $\phi^{-1}:\mbox{Im }\phi\to[1,+\infty)$ is well-defined on the image of $\phi$, and the expectation $\Exp{\phi(T)}$ remains in this image, thus ensuring the proper definition of $\mu_\phi$. For example, $\mu_{\mbox{\tiny Id}}(T)=\Exp{T}$ and $\mu_{\ln}(T)=e^{\Exp{\ln(T)}}$.
When $\phi$ is concave, Jensen's inequality immediately gives
\BEQ
\mu_\phi(T) \quad\leq\quad \Exp{T}\,,
\EEQ
and thus any algorithm achieving an expected running time $\LBtime{S} = O\left(\mu_\phi(T)\right)$ will also be, up to a multiplicative constant and in expectation, better than the original algorithm. These quantities are also often used for the analysis of tail probabilities of distributions, \eg with $\mu_{1-e^{-sx}}(T) = -\frac{1}{s}\ln\Exp{e^{-sT}}$ and Chernoff's concentration inequality.
We now show that the expected running time of $\SPRS$ and $\SSPRS$ are upper bounded by $O(\mu_\phi(T))$ for a large class of concave functions $\phi$. More specifically, we provide in \Sec{exactJ} a general condition for these upper bounds to hold. As this condition is hard to verify in practice, we then focus on the concave setting and provide in \Sec{sufficientJ} and \Sec{necessaryJ} near matching necessary and sufficient conditions for reverse Jensen's inequalities to hold: 1) $\LBtime{\SPRS} = O(\mu_\phi(T))$ if $\liminf_{t\to +\infty} \frac{t\phi''(t)}{\phi'(t)} > -2$, and 2) No restart strategy can reach $O(\mu_\phi(T))$ if $\limsup_{t\to +\infty} \frac{t\phi''(t)}{\phi'(t)} < -2$.

\subsection{Exact condition for reverse Jensen's inequalities}\label{sec:exactJ}

First, we provide a tight condition for quantities of the form $\inf_{q\in (0,1]}\frac{Q_p(q)}{\varphi(q)}$ to be upper bounded by quantities of the form $\mu_\phi(T)$.

\begin{lemma}\label{lem:iffbound}
Let $\varphi:(0,1]\to\R_+$ and $\phi:[1,+\infty)\to\R_+$ be two increasing and continuous functions, and, $\forall r\geq 1/\varphi(1)$, $c_{\phi,\varphi}(r) = \frac{1}{r}\phi^{-1}\left( \phi(r\varphi(1)) - \int_{\phi(1)}^{\phi(r\varphi(1))} \varphi^{-1}\left( \frac{\phi^{-1}(u)}{r} \right) du \right)$. Then, if $\liminf_{r\to +\infty} c_{\phi,\varphi}(r) > 0$, we have, for any distribution $p$,
\BEQ
\inf_{q\in (0,1]}\frac{Q_p(q)}{\varphi(q)} \quad\leq\quad C_{\phi,\varphi}\,\,\mu_\phi(T)\,,
\EEQ
where $C_{\phi,\varphi} = 1 / \inf_{r\geq 1/\varphi(1)} c_{\phi,\varphi}(r) < +\infty$. Moreover, this constant cannot be improved.
\end{lemma}
\begin{proof}
Let $r(p) = \inf_{q\in(0,1]} \frac{Q_p(q)}{\varphi(q)}$ and $f(p) = \frac{\mu_\phi(T_p)}{r(p)}$, where $T_p\sim p$. Our goal is to show that  $\inf_p f(p) > 0$.
For any $r\geq 1/\varphi(1)$, let $p_r$ be a probability distribution whose cumulative distribution function is $\Prob{T_{p_r}\leq \alpha} = \one\{\alpha\geq 1\}\min\{\varphi^{-1}(\alpha/r), 1\}$. By construction, we have that, $\forall\alpha\in[1,r\varphi(1)]$, $\varphi(\Prob{T_{p_r}\leq \alpha}) = \alpha/r$ and thus $r(p_r) = \inf_{\alpha\geq 1} \frac{\alpha}{\varphi(\Prob{T_{p_r}\leq \alpha})} = r$. Moreover, for any $p$ such that $r(p) = r$, $p$ stochastically dominates $p_r$, as $\forall \alpha\geq 1$, $\Prob{T_p \leq \alpha} \leq \varphi^{-1}(\alpha/r)$, and thus $\Prob{T_p \leq \alpha} \leq \Prob{T_{p_r} \leq \alpha}$. As a consequence, there exists $T_p\sim p$ and $T_{p_r}\sim p_r$ such that, almost surely, $T_p \geq T_{p_r}$, and thus $\mu_\phi(T_p) \geq \mu_\phi(T_{p_r})$. Combining these two inequalities, we obtain that, if $r(p)=r$, then $f(p) \geq f(p_r)$ and thus
\BEQ
\inf_p f(p) \quad=\quad \inf_{r \geq 1/\varphi(1)} f(p_r) \quad=\quad \inf_{r \geq 1/\varphi(1)} \frac{\mu_\phi(p_r)}{r}\,.
\EEQ
We conclude by computing the expectation $\Exp{\phi(T_{p_r})}$, as
\BEQ
\Exp{\phi(T_{p_r})} \quad=\quad \int_0^{+\infty} \Prob{\phi(T_{p_r}) > u} du \quad=\quad \phi(1) + \int_{\phi(1)}^{\phi(r\varphi(1))}\left( 1 - \varphi^{-1}\left( \frac{\phi^{-1}(u)}{r} \right) \right) du\,,
\EEQ
and thus $\mu_\phi(T_{p_r}) = rc_{\phi,\varphi}(r)$. Summarising, we showed that $\inf_p f(p) = \inf_{r \geq 1/\varphi(1)} c_{\phi,\varphi}(r)$, and we conclude by noting that $c_{\phi,\varphi}(r) \geq 1/r$, and thus $\inf_{r \geq 1/\varphi(1)} c_{\phi,\varphi}(r) = 0$ if and only if it approaches $0$ at infinity, \ie $\liminf_{r\to +\infty} c_{\phi,\varphi}(r) = 0$.
\end{proof}

Considering the upper bound of \Corollary{SAQ}, \Lemma{iffbound} shows that, for any $\phi$ satisfying the condition $\lim_{r\to +\infty} c_{\phi,\varphi}(r) > 0$ for $\varphi(x) = x/(1+\log(1/x)\log^2\log(1/x))$, we have
\BEQ
\LBtime{\SPRS} = O(\mu_\phi(T))\,.
\EEQ
Unfortunately, the condition $\liminf_{r\to +\infty} c_{\phi,\varphi}(r) > 0$ can be hard to prove for general functions $\varphi,\phi$. In order to simplify our analysis, we now consider functions $\varphi_\varepsilon(x) = x^{1+\varepsilon}$ and $\phi_\beta(x) = 1 - x^{-\beta}$.

\begin{corollary}\label{cor:powers}
Let $\varepsilon, \beta > 0$ be such that $\beta(1+\varepsilon) < 1$. Then, for any distribution $p$,
\BEQ
\inf_{q\in (0,1]}\frac{Q_p(q)}{q^{1+\varepsilon}} \quad\leq\quad C_{\beta,\varepsilon}\,\,\mu_{\phi_\beta}(T)\,,
\EEQ
where $\phi_\beta(x) = 1 - x^{-\beta}$ and $C_{\beta,\varepsilon} = (1-\beta(1+\varepsilon))^{-\frac{1}{\beta}}$.
\end{corollary}
\begin{proof}
We apply \Lemma{iffbound} to $\varphi_\varepsilon(x) = x^{1+\varepsilon}$ and $\phi_\beta(x) = 1 - x^{-\beta}$. A simple calculation gives
\BEQ
c_{\phi,\varphi}(r) = \frac{1}{r}\left[ r^{-\beta} + \int_0^{1-r^{-\beta}} r^{-\frac{1}{1+\varepsilon}} (1-u)^{-\frac{1}{\beta(1+\varepsilon)}} du \right]^{-\frac{1}{\beta}} = \left[ \frac{1}{1 - \beta(1+\varepsilon)} - \frac{\beta(1+\varepsilon)r^{\beta - \frac{1}{1+\varepsilon}}}{1 - \beta(1+\varepsilon)} \right]^{-\frac{1}{\beta}}\,,
\EEQ
and thus $\inf_{r\geq 1} c_{\phi,\varphi}(r) = (1 - \beta(1+\varepsilon))^{\frac{1}{\beta}}$.
\end{proof}

As $1+\log(1/q)\log^2\log(1/q) = O(q^{-\varepsilon})$ for any $\varepsilon\in(0,1)$, \Corollary{powers} immediately implies that
\BEQ
\LBtime{\SPRS} = O(\mu_{\phi_{1-\varepsilon}}(T))\,.
\EEQ

Note that these upper bound are much stronger than $O(\mu_{\ln}(T))$, as $\phi_{1-\varepsilon}$ are bounded functions, and thus $\mu_{\phi_{1-\varepsilon}}(T) < +\infty$ as soon as $\Prob{T = +\infty} < 1$, while $\mu_{\ln}(T) = +\infty$ as soon as $\Prob{T = +\infty} > 0$.

\subsection{Sufficient condition for reverse Jensen's inequalities}\label{sec:sufficientJ}
\Corollary{powers} can be extended to any concave function $\phi$ whose log-derivative $\ln\phi'$ decreases \emph{``more slowly''} than $-2/t$ (\ie the derivative of $\ln\phi_1'$) at infinity. This includes a large class of concave functions, and for example $\ln(x)$ and $\ln(1+\ln(x))$.

\begin{theorem}
Let $\phi:[1,+\infty)\to\R_+$ be a concave increasing and twice differentiable function such that $\liminf_{t\to +\infty} \frac{t\phi''(t)}{\phi'(t)} > -2$. Then, there exists a constant $C_\phi > 0$ such that, for any distribution $p$,
\BEQ
\LBtime{\SPRS} \quad\leq\quad C_\phi\,\,\mu_\phi(T)\,.
\EEQ
\end{theorem}
\begin{proof}
As $\liminf_{t\to +\infty} \frac{t\phi''(t)}{\phi'(t)} > -2$, there exists $\varepsilon_0 > 0$ and $B\geq 1$ such that, $\forall \varepsilon\in(0,\varepsilon_0)$ and $t\geq B$, $\frac{\phi''(t)}{\phi'(t)} \geq \frac{-2+\varepsilon}{t} = \frac{\phi_{1-\varepsilon}''(t)}{\phi_{1-\varepsilon}'(t)}$. As a consequence, $\phi\circ\phi_{1-\varepsilon}^{-1}$ is convex on $[\phi_{1-\varepsilon}(B),+\infty)$, and we will show that Jensen's inequality on a well chosen random variable will imply that $\liminf_{r\to +\infty} c_{\phi,\varphi_\varepsilon}(r) > 0$ and thus that \Lemma{iffbound} can be used.
First, as discussed in the proof of \Lemma{iffbound}, $c_{\phi,\varphi_\varepsilon}(r) = \mu_\phi(T_{p_r})/r$ where $T_{p_r}$ is a random variable such that, $\forall \alpha\geq 1$, $\Prob{T_{p_r}\leq \alpha} = \one\{\alpha\geq 1\}\min\{\varphi^{-1}(\alpha/r), 1\}$.
Moreover, we have, $\forall r\geq B$,
\BEQ
\BA{lll}
\Exp{\phi(T_{p_r})} &\geq& \Prob{T_{p_r}>B}\Exp{\phi(T_{p_r})~|~T_{p_r}>B}\\
&=& \Prob{T_{p_r}>B}\Exp{\phi\circ\phi_{1-\varepsilon}^{-1}(\phi_{1-\varepsilon}(T_{p_r}))~|~T_{p_r}>B}\\
&\geq& \Prob{T_{p_r}>B}\phi\circ\phi_{1-\varepsilon}^{-1}(\Exp{\phi_{1-\varepsilon}(T_{p_r})~|~T_{p_r}>B})\\
&\geq& \left(1 - \left(\frac{B}{r}\right)^{\frac{1}{1+\varepsilon}}\right)\,\phi(\mu_{\phi_{1-\varepsilon}}(T_{p_r}))\,,
\EA
\EEQ
where the second inequality is due to Jensen's inequality on $T_{p_r}$ conditionned on being superior to $B$. As discussed in the proof of \Corollary{powers}, a simple calculation gives $\mu_{\phi_{1-\varepsilon}}(T_{p_r}) \geq Cr$ where $C=\varepsilon^{\frac{2}{1-\varepsilon}}$ and thus
\BEQ
c_{\phi,\varphi_\varepsilon}(r) \quad\geq\quad \frac{1}{r}\phi^{-1}\left(\left(1 - \left(\frac{B}{r}\right)^{\frac{1}{1+\varepsilon}}\right)\,\phi\left(Cr\right)\right) \quad\geq\quad C - \frac{B^{\frac{1}{1+\varepsilon}}}{r^{1+\frac{1}{1+\varepsilon}}}\frac{\phi(Cr)}{\phi'(Cr)}\,,
\EEQ
by noting that $\phi^{-1}$ is convex and thus above its tangent at $\phi(Cr)$. To conclude, it is sufficient to show that there exists $\varepsilon\in(0,1)$ such that $\limsup_{t\to +\infty} \frac{\phi(t)}{t^{2-\varepsilon}\phi'(t)} = 0$, as this would imply $\liminf_{t\to +\infty} c_{\phi,\varphi_\varepsilon}(r) \geq C$.
For any $t\geq B$, let $g(t) = \frac{\phi'(t)}{\phi(t)}$ and $h(t)$ be the solution to the differential equation $h'(t) = -h(t)\left(\frac{2-\varepsilon_0}{t} + h(t)\right)$ with initial condition $h(B)=g(B)$. First, note that $\forall t \geq B$, $g'(t) = \frac{\phi''(t)\phi(t) - \phi'(t)^2}{\phi(t)^2} \geq -g(t)\left(\frac{2-\varepsilon_0}{t} + g(t)\right)$ and thus $g'(t) - h'(t) \geq F(h(t)) - F(g(t))$ where $F(u) = u\left(\frac{2-\varepsilon_0}{t} + u\right)$ is increasing for $u\geq 0$. As $g(t),h(t)\geq 0$, we have that $g'(t) \geq h'(t)$ if $g(t) \leq h(t)$, and thus $\forall t\geq B$, $g(t) \geq h(t)$. Moreover, $\frac{h'(t)}{h(t)} \leq \frac{-2+\varepsilon_0}{t}$ and integrating this inequality gives $h(t) \leq h(B)(t/B)^{-2+\varepsilon_0}$. Finally, using this inequality in the equation for $h'(t)$ gives
$\frac{h'(t)}{h(t)} \geq -\left(\frac{2-\varepsilon_0}{t} + h(B)(t/B)^{-2+\varepsilon_0}\right)$ that integrates to $h(t) \geq h(B)(t/B)^{-2+\varepsilon_0}e^{\frac{-h(B)B^{2-\varepsilon_0}}{1-\varepsilon_0}}$ as long as $\varepsilon_0 < 1$. Summarizing, we obtained that, $\forall t \geq B$, $g(t) \geq At^{-2+\varepsilon_0}$ where $A = h(B)B^{2-\varepsilon_0}e^{\frac{-h(B)B^{2-\varepsilon_0}}{1-\varepsilon_0}} > 0$, which immediately gives, when $\varepsilon\in(0,\varepsilon_0)$,
\BEQ
\limsup_{t\to +\infty} \frac{\phi(t)}{t^{2-\varepsilon}\phi'(t)} \quad\leq\quad \limsup_{t\to +\infty} \frac{t^{\varepsilon - \varepsilon_0}}{A} \quad=\quad 0\,,
\EEQ
concluding the proof.
\end{proof}

In particular, we have $\LBtime{\SPRS} = O\left(e^{\Exp{\ln(T)}}\right)$ and our strategy $\SPRS$ is at least as good as that of \cite{zamir2022wrong}. However, we can also obtain much stronger upper bounds, for example $\LBtime{\SPRS} = O\left(e^{e^{\Exp{\ln(1+\ln(T))}}}\right)$ or $\LBtime{\SPRS} = O\left( \Exp{T^a}^{\frac{1}{a}} \right)$ for any $a > -1$.

\subsection{Necessary condition for reverse Jensen's inequalities}\label{sec:necessaryJ}
Conversely, we now show that reverse Jensen's inequalities cannot be obtained for functions whose log-derivative $\ln\phi'$ decreases \emph{``faster''} than $-2/t$.
First, we show that $\mu_{\phi_1}(T)$ where $\phi_1(x) = 1-\frac{1}{x}$ is a lower bound for the optimal expected running time $\ell^\star_p$, and is thus the smallest such reverse Jensen's inequality that can be hoped for.

\begin{lemma}\label{lem:lb}
Let $\phi_\beta(x) = 1 - x^{-\beta}$. Then, $\mu_{\phi_1}(T) = \frac{1}{\Exp{1/T}}$ and
\BEQ
\mu_{\phi_1}(T) \quad\leq\quad \inf_{S\in[1,+\infty)^\N} \LBtime{S} \quad\leq\quad e^2\,\mu_{\phi_1}(T)\,(1+\ln\mu_{\phi_1}(T))\,.
\EEQ
\end{lemma}
\begin{proof}
The first inequality is a consequence of, $\forall \alpha\geq 1$,
\BEQ
\Exp{\one\{T\leq\alpha\}} \quad\leq\quad \Exp{\min\left\{1,\frac{\alpha}{T}\right\}} \quad=\quad \Exp{T^{-1}\,\min\{T,\alpha\}} \quad\leq\quad \Exp{T^{-1}}\Exp{\min\{T,\alpha\}}\,,
\EEQ
where the first inequality come from $\one\{t\leq\alpha\} \leq \min\{1,\alpha/t\}$, while the second inequality is due to the negative correlation of $T^{-1}$ and $\min\{T,\alpha\}$. Thus, we have $\mu_{\phi_1}(T) \leq \min_{\alpha \geq 1} \frac{\Exp{\min\{T,\alpha\}}}{\Prob{T\leq\alpha}} = \inf_{S\in[1,+\infty)^\N} \LBtime{S}$.
The second inequality is a consequence of \Corollary{powers} with $\beta = (1+a)^{-1}$ and $\varepsilon=0$, and noting that, as $1/T\leq 1$, we have $\mu_{\phi_{1/(1+a)}}(T) \leq \mu_{\phi_1}(T)^{1+a}$. This gives $\inf_{S\in[1,+\infty)^\N} \LBtime{S} \leq \left(1+\frac{1}{a}\right)^{1+a}\,\mu_{\phi_1}(T)^{1+a}$ and setting $a = 1/\ln\mu_{\phi_1}(T)$ gives the desired result.
\end{proof}

Intuitively, \Lemma{lb} shows that $\mu_{\phi_1}(T)$ lower bounds the best expected running time of restart strategies, and thus no restart strategy can achieve an expected running time significantly better than $\mu_{\phi_1}(T)$.
As a consequence, we now show that achieving a reverse Jensen's inequality $\LBtime{S} = O(\mu_\phi(T))$ for a function $\phi$ such that $\ln\phi'$ decreases \emph{``faster''} than $-2/t$ would contradict the lower bound in $\ell^\star_p\log \ell^\star_p$ of \Theorem{luby}.
\begin{theorem}\label{th:lb2}
Let $\phi:[1,+\infty)\to\R_+$ be an increasing and twice differentiable function such that $\limsup_{t\to +\infty} \frac{t\phi''(t)}{\phi'(t)} < -2$. Then, there exists no sequence $S\in[1,+\infty)^\N$ and constant $C>0$ such that, for any distribution $p$, $\LBtime{S} \leq C\,\mu_\phi(T)$.
\end{theorem}
\begin{proof}
Let us assume that there exists $S$ such that $\LBtime{S} \leq O(\mu_\phi(T))$. As $\limsup_{t\to +\infty} \frac{t\phi''(t)}{\phi'(t)} < -2$, there exists $B > 0$ such that $\forall t\geq B$, $\frac{t\phi''(t)}{\phi'(t)} \leq -2$ and thus $\phi\circ\phi_1^{-1}$ is concave on $[\phi_1(B),+\infty)$. As $T+B \geq B$, applying Jensen's inequality to $\phi\circ\phi_1^{-1}$ thus gives
\BEQ
\mu_\phi(T) \quad\leq\quad \mu_\phi(T + B) \quad\leq\quad \mu_{\phi_1}(T + B)\,.
\EEQ
As a consequence, \Lemma{lb} implies that $\LBtime{S} = O(\mu_{\phi_1}(T+B)) = O(\ell^\star_{\tilde{p}_B})$ where $\tilde{p}_B$ is the distribution of the random variable $T+B$. Finally, we conclude by noting that $\ell^\star_{\tilde{p}_B} \leq \min_{q\in(0,1]} \frac{Q_p(q) + B}{q} \leq 2(1+B)\ell^\star_p$ as $Q_p(q) \geq 1$. However, this implies that $\LBtime{S} = O(\ell^\star_p)$ which contradicts the lower bound of \Theorem{luby} in $\Omega(\ell^\star_p\ln\ell^\star_p)$.
\end{proof}

For example, \Theorem{lb2} implies that it is impossible to obtain an expected running time in $O(\mu_\phi(T))$ for $\phi(x) = 1 - x^{-\beta}$ and $\beta\geq 1$, or for $\phi(x) = 1 - e^{-sx}$ and $s > 0$, and shows that the behavior of $\frac{t\phi''(t)}{\phi'(t)}$ at infinity controls the existence or absence of an algorithmic reverse Jensen's inequality for any concave function $\phi$.

\section{Conclusion}
In this work, we showed that the universal strategy of \cite{luby1993optimal} can be further improved by using simulated parallel computations of restart strategies achieving fixed quantiles. Our novel restart strategies $\SPRS$ and $\SSPRS$ show superior performance compared to $\Suniv$ of \cite{luby1993optimal} in low variance settings in which the running time is close to deterministic, or when the optimal quantile $Q_p(q^\star)$ is large compared to $1/q^\star$. Moreover, we showed that these algorithms exhibit upper bounds of the form $O(\mu_\phi(T))$ where $\mu_\phi(T)=\phi^{-1}(\Exp{\phi(T)})$ for a large class of concave functions including powers $\phi(x) = x^a$ for $a\in(0,1)$, iterated logarithms $\phi(x) = \ln x$ and $\phi(x) = \ln(1+\ln x)$, and even bounded functions of the form $\phi(x) = 1 - x^{-\beta}$ for $\beta\in(0,1)$.
Further work includes better pinpointing the conditions on $\phi$ for which the reverse Jensen's inequality $\LBtime{\SPRS} = O(\mu_\phi(T))$ holds, as well as investigating the optimality of SPRS, for example by providing a matching lower bound on the expected running time.

\section*{Acknowledgements}
The author would like to thank Marc Lelarge for bringing this problem to his attention, and Ana Busic, Laurent Massoulié and the whole Argo team at Inria Paris for providing valuable feedback on the work and engaging in passionate discussions.

\newpage
\bibliographystyle{alpha}
\bibliography{bibliography}

\appendix

\end{document}